\documentclass{article}
\usepackage{arxiv}

\usepackage{graphicx}     

\usepackage[square, numbers]{natbib}

\usepackage{amsmath} 
\usepackage{amssymb}  

\usepackage{amsthm}
\usepackage{color,soul}
\usepackage{enumitem}
\usepackage{flushend}
\usepackage{cite}
\usepackage{url}
\usepackage{mathrsfs}  

\usepackage{algorithm} 
\usepackage{algorithmic}  
\usepackage[linesnumbered,algo2e,ruled,vlined,norelsize]{algorithm2e} 

\allowdisplaybreaks

\newtheorem{thmm}{Theorem}

\newtheorem{remm}{Remark}
\newtheorem{assm}{Assumption}

\newtheorem{Lemm}{Lemma}

\title{Merging Large Language Models and Battery Physics for User-Aware Electric Vehicle Driving Management} 
\author{Yukta Pareek\\
Department of Mechanical Engineering\\ The Pennsylvania State University\\ University Park, Pennsylvania 16802, USA\\E-mail: ybp5153@psu.edu.\\	
\And
Yasaman Masoudi\\
Stellantis US LLC\\  Auburn Hills, Michigan 48326, USA\\E-mail: yasaman.masoudi@stellantis.com.\\		
\And
Satadru Dey\\
Department of Mechanical Engineering\\ The Pennsylvania State University\\ University Park, Pennsylvania 16802, USA\\E-mail: skd5685@psu.edu.\\}

\begin{document}

\maketitle







\begin{abstract}                
Electric vehicle (EV) battery performance is strongly coupled with driver behavior, yet human intent is typically expressed semantically rather than numerically. This paper proposes a hybrid physics–artificial intelligence framework that integrates a Large Language Model (LLM) as a high-level behavioral reasoning layer within a physics-driven supervisory architecture. The LLM interprets textual user intent and structured battery feedback to generate bounded behavioral parameters that shape a discharge current envelope. A physics-driven safety filter then enforces physical safety constraints before computing feasible velocity recommendations. Lyapunov-based analysis establishes bounded recommendation error under battery model and prompt inaccuracies. Simulation results demonstrate adaptive, user-aware operation without compromising physical safety. The proposed reasoning–enforcement architecture provides a principled pathway for safe AI integration in EV energy management.
\end{abstract}


\section{Introduction}





Most modern EVs rely on battery packs as the main power source whose power capability, thermal stability, and degradation dynamics are tightly coupled to real-time usage/driving patterns. Driver behavior, such as acceleration intensity, cruising speed, and transient power demand, directly influences battery current peaks, internal temperature rise, and long-term state-of-health evolution. Consequently, EV operation can be viewed as an interaction between human intent and physical constraints of the battery pack. Incorporating user awareness into EV management systems is therefore critical for achieving safe, efficient, and longevity-preserving operation. 


In connection with user-awareness and EV operation, there is a substantial literature on EV charging behavior -- factors influencing when, where, and how users charge (price, convenience, range anxiety, time-of-day), and the implications for grid operation and charging infrastructure planning \citep{fotouhi2019general,franke2013understanding, hu2019modeling}. These works typically focus on charging choice and not on the driving behaviour (acceleration/braking patterns) that directly affects instantaneous battery power and thermal stress. Separate from charging studies, decades of traffic research provide microscopic driver models (car-following, lane-changing) such as the Intelligent Driver Model (IDM)\citep{alhariqi2022calibration} and MOBIL for lane changes \citep{kesting2007general}. These models capture acceleration/deceleration behavior and inter-vehicle interactions and are widely used in traffic micro-simulation. IDM and its variants are useful to generate plausible velocity/acceleration time series, which can be mapped to instantaneous power demands for vehicle energy consumption studies.
Furthermore, a key challenge in capturing user-awareness in driving behavior is that human intent, urgency, and preference are often expressed semantically rather than numerically, making them difficult to encode in compact mathematical form. Effective EV supervision therefore requires integration of multi-modal information: semantic user descriptions and numerical battery-vehicle states. Since traditional control architectures are not designed to fuse such heterogeneous modalities, we explore the effectiveness of LLMs which can be suitable in handling such multi-modal information.

\begin{figure*}[h]
    \centering    
    \includegraphics[width=\textwidth]{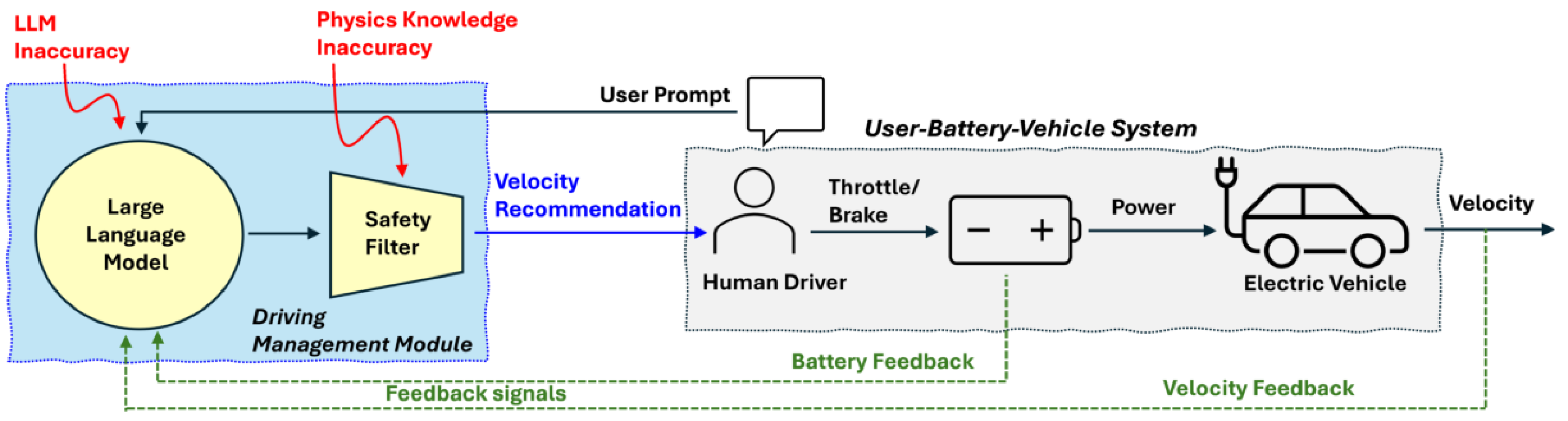}
    \caption{Schematic of EV driving management system.}
    \label{block_diagram}
\end{figure*}

LLMs are transformer-based sequence models trained via large-scale next-token prediction \citep{vaswani2017attention}. Through attention mechanisms, they encode semantic structure \citep{chen2024evaluating} and exhibit emergent reasoning capabilities under appropriate prompting \citep{wei2022chain}. However, a key limitation of LLMs is that they do not inherently encode physical conservation laws, stability constraints, or electrochemical dynamics. Prior analyses have noted that purely language-trained models may lack a consistent physical world-model, potentially producing physically implausible outputs when extrapolating beyond training distributions \citep{marcus2019rebooting}. Accordingly, safety-critical deployment requires embedding LLM outputs within physics-driven constraint enforcement.

In summary, the existing literature suffers from one or more of the following limitations: (i) Some works consider EV drivers' charging behavior, however, EV driving behavior remain underexplored. (ii) Most of these works do not consider the coupled dynamics of EV user and battery -- which is critical in understanding the broader dynamical effects experienced by battery systems. In light of the aforementioned research gaps, the main contribution of this work is as follows: \textit{We propose a framework for EV driving  management which captures the coupled dynamical nature of EV user-battery system -- unifying battery physics (via physics-based models) and human driver dynamics (via LLM-driven user models).} Specifically, we present a driving management system where the driver feeds an initial user prompt to an LLM to generate driving recommendation which is subsequently passed through a physics-driven safety filter to ensure safe recommendations. The LLM receives feedback of the driver's velocity tracking performance and battery states, and periodically updates driving recommendation. Using Lyapunov's stability theory, we also analyze the bounds of recommendation errors from this driving management system, which could potentially arise from inaccurate knowledge of the user and imperfect knowledge of battery physics. The remainder of the paper is structured as follows: Section 2 presents the driving management framework, Section 3 analyzes potential recommendation errors of the proposed framework, Section 4 presents simulation results, and Section 5 concludes the work. \textbf{Notation:} The subscript $k$ indicates discrete $k$-th time instant and $\left\| . \right\|$ indicates the $L_2$ vector norm.

\section{Driving Management Framework}

A schematic of the proposed driving management module is shown in Fig. \ref{block_diagram}. A brief working principle of the module is given as follows: In the beginning of the trip, the driver provides LLM an initial prompt of their plan. Based on this prompt, driver usage history, and battery states information, LLM suggests a battery current. This recommended current then passes through a safety filter to ensure that such current application leads to a safe operation in terms of battery dynamics. If the recommended current does not pass this safety check, the recommended current is scaled down until the scaled version passes safety check. After that, the safe current is used to simulate vehicle dynamical model to generate a corresponding safe velocity -- which is then provided to the driver as the recommended velocity. This process of is updated periodically. That is, the LLM receives driver's velocity tracking data as well as real-time battery states and in turn generates recommended battery current for each period. 

\vspace{2mm}

\begin{remm}
    This approach essentially constitutes a hybrid physics-Artificial Intelligence framework for user- and battery-aware operation of EVs. Here, an LLM is integrated as a high-level behavioral reasoning layer, while all real-time safety and dynamics are enforced by lower-level physics-based filters. The key design principle is a strict separation between reasoning and enforcement. LLM outputs a battery current envelope that reflects user behavior and recent battery response. A physics-based safety filter then generates and execute feasible commands within this envelope. This separation ensures safety, interpretability, and robustness while enabling contextual adaptation that is difficult to encode using fixed rules.
\end{remm}

\subsection{Mathematical Representation of User-Battery-Vehicle}

We consider a longitudinal dynamical model for vehicle motion \citep{eriksson2014modeling}:
\begin{align}
    {v}_{k+1} = v_{k} +K_1 F_{{drive}_k} - K_2v_k^2-K_3, \label{vd} 
\end{align}
where $v$ is the longitudinal velocity (m/s) and $m$ is the mass (kg); $K_1 = \Delta t \frac{1}{m}, K_2 = \frac{K_1}{2}C_dA_f\rho_{a}, K_3 = K_1C_rmg$ with $\Delta t$ being the sample time; $F_{drive}$ is the powertrain force (N) due to the user's driving command; the second and third terms on the right hand side represent the aerodynamic drag force and rolling resistance force, respectively; and $C_d,A_f,\rho_{a},C_r,g$ are the drag coefficient, frontal area, air density, rolling coefficient, and gravitational constant, respectively. The powertrain force $F_{drive}$ is generated from the vehicular battery pack, whose dynamics is modelled by the approximated electrochemical Single Particle Model (SPM) \citep{santhanagopalan2006online,dey2015nonlinear}, as given below:
\begin{align}
    {b}_{k+1} = A b_k + B I_k, \ P_k = g(b_k,I_k),\label{ssbatt} 
\end{align}
where $b$ is the Lithium concentration vector (mol/m$^3$) in a single cell, $I$ is the cell current (A), $P$ is the pack power that captures the cell-to-pack conversion due to series-parallel cell arrangements in the pack, the system matrices $A$, $B$, and the function $g(.)$ are formed by approximating the SPM as discussed in \citep{dey2015nonlinear}. 

\subsection{LLM-Based User-Aware Recommendation System}
The vehicle and battery dynamics introduced in Section 2.1 describes how physical quantities evolve once a driving command is applied. However, they do not specify how such a command should be generated in a user-aware manner. In this work, we introduce LLM as a high-level reasoning module that interprets human intent, takes battery feedback and translates it into physically meaningful control recommendations. 
At the beginning of a trip, the driver provides a semantic description of intent (e.g., “30-minute city drive,” “running late”). This textual input is not directly convertible into numerical control parameters using classical rule-based controllers. Instead, it is processed by the LLM, which interprets user intent jointly with numerical battery feedback and generates two interpretable behavioral parameters. For our analysis, we will use the following abstraction of LLM:
\begin{align}
    [\alpha_k,\beta_k]^T = f_{LLM}(\theta_k, T_k, \mathcal{B}_k), \label{llm}
\end{align}
where $\alpha,\beta$ are the LLM outputs, $\theta$ represents the internal parameterization of LLM, $T$ is the vector representation of user prompt, $\mathcal{B}$ is the periodic battery feedback supplied to the LLM, and $f_{LLM}(.)$ is the functional representation of LLM. Note that the actual text sequences are converted to vector representations through some embedding functions in LLMs \citep{tang2024understanding}. Parameter $\alpha \in [0,  \alpha_{max}]$ scales the base discharge capability of the battery and captures user-driven aggressiveness in acceleration demand. Larger values of $\alpha$ correspond to higher permissible current envelopes. Parameter $\beta \in [0, 1]$ captures behavioral uncertainty, and is used to tighten the current envelope to ensure robustness under ambiguous conditions. The LLM output $[\alpha_k,\beta_k]$ is translated into a current envelope using a  deterministic mapping: 
\begin{align}
    &{I}_{r_k} = h_{LP}({\alpha}_k,{\beta}_k) = {\alpha}_k I_{b}(1-\kappa {\beta}_k), \label{I_LLM_calculation}
\end{align}
 where $I_{b}$ is a nominal pack discharge limit and $\kappa$ is an user-defined parameter. The current envelope $I_{r}$ is then passed to the safety filter block, which enforces all physical constraints before generating a velocity recommendation described next.

\subsection{Physics-Driven Safety Filter}

The safety filter acts as the enforcement layer between the LLM reasoning module and the physical battery-vehicle system as described in Section 2.1 and 2.2. Its primary objective is to ensure that all recommended actions satisfy voltage constraints derived from the battery model \eqref{ssbatt}. The safety filter computes a feasible current command.
\begin{align}
    I_{s_k} = \gamma(V_k)I_{r_k}, \label{I_safe calculation_1} 
\end{align}

where $\gamma(V_k) \in [0,1 ]$ is a continuous scaling function such that
$\gamma(V_k) = 1$ when voltage margin is sufficient, and $\gamma(V_k) \to 0$ as $V_k \to V_{min}$. Next, the filter computes the feasible mechanical power as
\begin{align}
    P_{k} = g(V_k,I_k), \label{P_mech_safe_1}
\end{align}
where the function $g(.)$ depends on drivetrain efficiency $\eta$. Using $F_{{drive}_k}=P_{k}/{v}_{r_{k}}$ the longitudinal vehicle model \eqref{vd} the recommended velocity $v_r$ is computed.
The safety filter therefore determines the recommended velocity command which is provided to the driver, completing one supervisory cycle.

\section{Analysis of Recommendation Errors}

In this section, we analyze the performance of the proposed driving management system in terms of error between \textit{ideal} recommendation and \textit{actual} recommendation. We consider two major sources that contribute to the difference between these two quantities: 
\begin{itemize}
    \item \textit{Inaccurate battery model estimate} which arises from imperfect knowledge of physics.
    \item \textit{Inaccurate user behavior estimate} which arises from the imperfect prompting and/or inaccurate usage history provided to the LLM.
\end{itemize}
We frame this analysis into the following problem:

\noindent \textbf{Problem Statement:} Given the \textit{inaccurate battery model estimate}, quantified by the deviations in system matrices $\delta_A$ and $\delta_B$ in \eqref{ssbatt}, and the \textit{inaccurate user behavior estimate}, quantified by the deviations in LLM parameterization $\delta_\theta$ and deviations in user prompt $\delta_T$ in \eqref{llm} -- what is the upper bound of the error between \textit{ideal} recommendation ($\bar{v}_{{r}}$) and \textit{actual} recommendation ($v_{r}$)?

Consider the \textit{ideal} recommendation ($\bar{v}_{{r}}$) is generated through the following sets of algorithms:
\begin{align}
    & [\bar{\alpha}_k,\bar{\beta}_k]^T = f_{LLM}(\theta_k+\delta_\theta, T_k+\delta_T, \mathcal{B}_k), \label{lp-21} \\
    & \bar{I}_{r_k} = h_{LP}(\bar{\alpha}_k,\bar{\beta}_k), \label{lp-22}\\
    & {\bar{b}}_{k+1} = (A+\delta_A) \bar{b}_{k} + (B+\delta_B) \bar{I}_{r_k}, \ \bar{P}_k = g(\bar{b}_k,\bar{I}_{r_k}) \label{lp-23}\\
    & \bar{v}_{r_{k+1}} = \bar{v}_{r_{k}} +K_1 \bar{P}_k/\bar{v}_{r_{k}} - K_2\bar{v}_{r_{k}}^2-K_3 \label{lp-24} 
\end{align}

Next, consider the \textit{actual} recommendation (${v}_{{r}}$) is generated through the following sets of algorithms:
\begin{align}
    & [{\alpha}_k,{\beta}_k]^T = f_{LLM}(\theta_k, T_k, \mathcal{B}_k), \label{lp-31} \\
    & {I}_{r_k} = h_{LP}({\alpha}_k,{\beta}_k), \label{lp-32}\\
    & {{b}}_{k+1} = A {b}_{k} + B {I}_{r_k}, \ {P}_k = g({b}_k, {I}_{r_k})\label{lp-33} \\
    & {v}_{r_{k+1}} = {v}_{r_{k}} + K_1 {P}_k/{v}_{r_{k}} - K_2{v}_{r_{k}}^2-K_3. \label{lp-34}
\end{align}

Our goal is to analyze the bound on the velocity recommendation error $e_{v_k} = \bar{v}_{{r}_k} - {v}_{{r}_k}$. We will denote the error between the battery states under these two cases, that is, $e_{b_k} = \bar{b}_{k}-{b}_{k}$. Considering the dynamical nature of the overall \textit{user-battery-vehicle} system as well as the inherent dynamics of \textit{LLM-Safety Filter} recommendation system, we leverage Lyapunov stability theory for discrete-time systems \citep{haddad2008nonlinear} to perform the above-mentioned error analysis.

\vspace{2mm}

\begin{assm}
The composite function $hf (.) = h_{LP} \circ f_{LLM}(.)$ in \eqref{lp-21}, \eqref{lp-22}, \eqref{lp-31}, and \eqref{lp-32} is Lipschitz continuous, as described below:
    \begin{align}
        & \left\| hf(\lambda_k + \delta_{\lambda}, \mathcal{B}_k) -  hf(\lambda_k, \mathcal{B}_k) \right\| \leqslant L_c \left\| \delta_{\lambda} \right\|, \label{lip-1}
    \end{align}
    where $\lambda_k = [\theta_k,T_k]^T$, $\delta_{\lambda} = [\delta_\theta,\delta_T]^T$, and $L_c = L_{HP}L_{\lambda}$ with $L_{\lambda}$ and $L_{HP}$ being the Lipschitz constants for the functions $f_{LLM}(.)$ and $h_{LP}(.)$, respectively.
\end{assm}

\vspace{2mm}

\begin{assm}
The function $g (.)$ in \eqref{lp-23} and \eqref{lp-33} is Lipschitz continuous, as described below:
    \begin{align}
        & \left\|  g(\bar{b}_{k},\bar{I}_{r_k}) -  g({b}_{k},{I}_{r_k})\right\| \leqslant L_1 \left\|\bar{b}_k- {b}_k\right\| + L_2 \left\|\bar{I}_{r_k}- {I}_{r_k}\right\|, \label{lip-2}
    \end{align}
    where $L_1$ and $L_2$ are the Lipschitz constants.
\end{assm}

\vspace{2mm}

\begin{remm}
The Lipschitz continuity of the composite function $hf(.) = h_{LP} \circ f_{LLM}(.)$ depends on the Lipschitz continuity of individual functions. Lipschitz continuity of $h_{LP}(.)$ can be ensured by choosing a smooth conversion function from LLM outputs to recommended current. For the LLM function $f_{LLM}(.)$, {Lipschitz continuity can be justified under standard neural network regularity assumptions. Transformer architectures consist of affine mappings, softmax attention operators, and smooth nonlinear activations \citep{virmaux2018lipschitz}. Each of these components is Lipschitz continuous over bounded domains.} The Lipschitz continuity of battery power $g (.)$ can be confirmed following the approach mentioned in \citep{dey2015nonlinear}, given the smoothness of battery voltage expression.   
\end{remm}

\vspace{2mm}

\begin{Lemm}
If Assumptions 1 and 2 are true, the function $\tilde{P}_k \triangleq \bar{P}_k -P_k$ is upper bounded by the following: $\left\| \tilde{P}_k \right\| \leqslant L_1 \left\|e_{b_{k}}\right\| + L_2 L_c \left\| \delta_{\lambda} \right\|.$
\end{Lemm}

\begin{proof}
Considering \eqref{lp-32} and \eqref{lp-22} and Assumption 1, we can see that $\bar{I}_{r_k}- {I}_{r_k} = hf(\lambda_k + \delta_{\lambda}, \mathcal{B}_k) -  hf(\lambda_k, \mathcal{B}_k)$. Then, using \eqref{lip-1}, we can write $\left\|\bar{I}_{r_k}- {I}_{r_k}\right\| \leqslant L_c \left\| \delta_{\lambda} \right\|$. Next, considering \eqref{lp-33} and \eqref{lp-23}, we can write $\tilde{P}_k = g(\bar{b}_k,\bar{I}_{r_k}) - g({b}_k, {I}_{r_k})$. Then, considering Assumption 2 and then Assumption 1, we can arrive at the upper bound.
\end{proof}

\vspace{2mm}

\begin{thmm}
Consider the \textit{ideal} recommendation ($\bar{v}_{{r}}$) generated by \eqref{lp-21}-\eqref{lp-24} and \textit{actual} recommendation (${v}_{{r}}$) generated by \eqref{lp-31}-\eqref{lp-34}. If Assumptions 1 and 2 are true, the velocity recommendation error $e_{v_{k}}$ is bounded by battery model inaccuracy $\delta_b$ and LLM inaccuracy $\delta_\lambda$, in the following sense: 
\begin{align}
    & \left\| e_{v_{k}} \right\|^2 \! \leqslant \! \phi_{k,0} \left\| e_{v_{0}} \right\|^2 \!+\!
     \textstyle\sum_{j=0}^{k-1} \phi_{k,j+1} \{L_1 (\gamma^k \left\| e_{b_0} \right\|^2 + \nonumber \\
     &   \frac{3}{1\!-\!\gamma} ( \left\| \delta_b \right\|^2 \left\| \bar{m} \right\|^2  \!+\! \left\| B \right\|^2 L^2_c \left\| \delta_{\lambda} \right\|^2))^{0.5} \!+\! L_2 L_c \left\| \delta_{\lambda} \right\|\}^2, \label{bnd}
\end{align}
where $\gamma = \lambda_{max}(A^TA-I)+2 \left\|A\right\|^2+1$ with $\lambda_{max}(.)$ representing the maximum eigen-value, and $\phi_{k,0} = \prod_{j=0}^{k-1} \frac{2 {\gamma^2_1}_k}{1-{\lambda_2}_j}$ with ${\gamma_1}_k = K_1/{v}_{{r}_k}$, ${\lambda_2}_k=(-1+2{\lambda^2_1}_k)$, ${\lambda_1}_k = (1-K_2(\bar{v}_{{r}_k} + {v}_{{r}_k})-K_1\bar{P}_r/(\bar{v}_{{r}_k}{v}_{{r}_k}))$.
\end{thmm}

\begin{proof}
First, we consider the dynamics of the error $e_{b_k} = \bar{b}_{k}-{b}_{k}$.
 Subtracting \eqref{lp-33} from \eqref{lp-23}, we formulate the battery state error difference equation as:
\begin{align}
    e_{b_{k+1}} & = {\bar{b}}_{k+1}-{{b}}_{k+1} = A e_{b_{k}} + \delta_b m_k + B (\tilde{f}\tilde{h})
     \label{ly-1}
\end{align}
where $\delta_b = [\delta_A, \delta_B]$, $m_k = [\bar{b}_{k}, \bar{I}_{r_k}]^T$ and $\tilde{f}\tilde{h} = hf(\theta_k+\delta_\theta, P_k+\delta_P, \mathcal{B}_k) - hf(\theta_k, P_k, \mathcal{B}_k )$

Considering the Lyapunov function candidate $W_{b_k} = \left\| e_{b_{k}} \right\|^2$, we can write:
\begin{align}
    & \Delta W_b = W_{b_{k+1}} - W_{b_k} = (\lambda_{max}(A^TA-I)+2 \left\|A\right\|^2) \left\| e_{b_{k}} \right\|^2 \nonumber\\
    & + 3 \left\| \delta_b \right\|^2 \left\| \bar{m} \right\|^2 + 3\left\| B \right\|^2 L^2_c \left\| \delta_{\lambda} \right\|^2, \label{int-1}
\end{align}
Under the condition 
\begin{align} 
& \left\| e_{b_{k}} \right\|^2 \geqslant \frac{3 (\left\| \delta_b \right\|^2 \left\| \bar{m} \right\|^2 
+ 3\left\| B \right\|^2 L^2_c \left\| \delta_{\lambda}\right\|^2)} {-(\lambda_{max}(A^TA-I)+2 \left\|A\right\|^2)},
\end{align}
we have $\Delta W_b \leqslant 0$ which ensures that $\left\| e_{b_{k}} \right\|$ will be bounded. From \eqref{int-1}, we can write
\begin{align}    
    & {W_b}_{k+1} \leqslant \gamma W_{b_k} + 3 \left\| \delta_b \right\|^2 \left\| \bar{m} \right\|^2 + 3\left\| B \right\|^2 L^2_c \left\| \delta_{\lambda} \right\|^2, \label{lya-22}
\end{align}
The difference inequality \eqref{lya-22} results in
\begin{align}
    & {W_b}_{k} \leqslant \gamma^k W_{b_0} +  \frac{3}{1-\gamma} ( \left\| \delta_b \right\|^2 \left\| \bar{m} \right\|^2 + \left\| B \right\|^2 L^2_c \left\| \delta_{\lambda} \right\|^2), \nonumber
\end{align}
from which we can write
\begin{align}
    & \left\| {e_b}_{k} \right\|^2 \leqslant \gamma^k \left\| e_{b_0} \right\|^2 +  \frac{3}{1-\gamma} ( \left\| \delta_b \right\|^2 \left\| \bar{m} \right\|^2 + \left\| B \right\|^2 L^2_c \left\| \delta_{\lambda} \right\|^2), \label{abc-1}
\end{align}
Next, we write the error dynamics of $e_{v_{k}}=\bar{v}_{{r}_k} - {v}_{{r}_k}$.
\begin{align}
    e_{v_{k+1}} & = \lambda_1 e_{v_{k}} + \gamma_1 \tilde{P}_k 
     \label{ly-222}
\end{align}
Considering the Lyapunov function candidate $W_{v_k} = \left\| e_{v_{k}} \right\|^2$, we can write:
\begin{align}
     \Delta W_v = W_{v_{k+1}} - W_{v_k} & = \left\| e_{v_{k+1}} \right\|^2 - \left\| e_{v_{k}} \right\|^2 \nonumber\\ 
    & = -\lambda_2 e^2_{v_{k}} + 2 \gamma^2_1 \tilde{P}^2_k, \label{ab-1}
\end{align}
Under the condition $\left\| e_{v_{k}} \right\| \geqslant (\sqrt{2}\gamma_1/\lambda_2) \left\| \tilde{P}_k \right\|$, we have $\Delta W_v \leqslant 0$ which ensures that $\left\| e_{v_{k}} \right\|$ will be bounded. From \eqref{ab-1} and Lemma 1, we can further write:
\begin{align}
    & W_{v_{k+1}} \leqslant (1-\lambda_2) W_{v_k} + 2 \gamma^2_1 \{L_1 \left\|e_{b_{k}}\right\| + L_2 L_c \left\| \delta_{\lambda} \right\|\}^2, \label{ab-2}
\end{align}
The difference inequality \eqref{ab-2} will lead to
\begin{align}
    & W_{v_k} \!\leqslant \! \phi_{k,0} W_{v_0} \!+\! \textstyle\sum_{j=0}^{k-1} \phi_{k,j+1} \{L_1 \left\|e_{b_{j}}\right\| + L_2 L_c \left\| \delta_{\lambda} \right\|\}^2, \label{ab-3111}
\end{align}
Using \eqref{abc-1}, \eqref{ab-3111} can be re-written as \eqref{bnd}.
\end{proof}

\begin{figure*}[ht]
    \centering    
    \includegraphics[width=\textwidth,trim={3.5cm 0 1.7cm 0},clip]{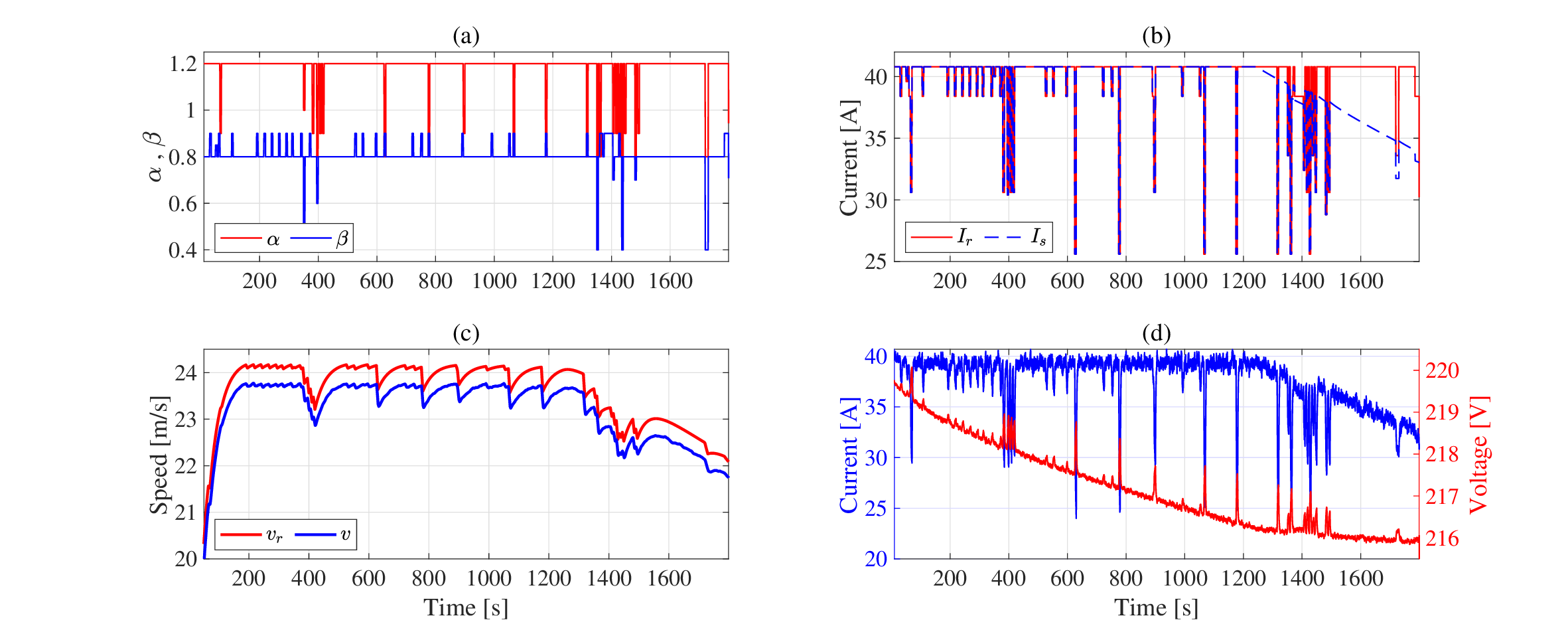}
    \caption{Closed-loop performance of the proposed LLM–physics driving management framework for a 30 minutes drive -- with LLM getting battery feedback and recommending $\alpha, \beta$ every 5 seconds. (a): Evolution of LLM-generated behavioral parameters $\alpha$ and $\beta$. (b): Recommended current $I_r$ from LLM and safety filtered current $I_s$. (c): Actual vehicle velocity $v$ tracking the recommended velocity $v_r$. (d): EV battery current and voltage trajectories. }
    \label{results2}
\end{figure*}

\section{Results and Discussion}


This section evaluates the performance and robustness of the proposed driving management framework. We consider a driving scenario in which the LLM adapts the discharge envelope every 5 seconds. The implementation follows a co-simulation architecture in which the LLM reasoning layer is executed in Python 3.9.6, while the vehicle dynamics, battery electrochemical model, safety filter, and closed-loop enforcement layers are implemented in MATLAB 2024b. We present two sets of results: performance under adaptive LLM recommendations, and robustness analysis under LLM prompt and battery model inaccuracies.

\noindent\textbf{Performance under adaptive LLM recommendations:} To evaluate interaction between semantic reasoning and physics-based constraint enforcement, LLM was provided with multi-modal inputs consisting of (i) structured numerical battery summaries (from MATLAB) and (ii) textual descriptions of user intent. The numerical summary included minimum pack voltage, mean and peak discharge current, instantaneous and mean speed, and state-of-charge. These values were embedded directly into the prompt at each supervisory update interval. The LLM instruction also specifies to output only valid JSON for $\alpha$ and $\beta$. In Fig. \ref{results2}(a), LLM-generated behavioral parameters $\alpha_k$, $\beta_k$ are shown over time. The LLM receives recent battery feedback -- including peak current and voltage margin -- and it is observed that LLM proactively reduces $\alpha$ and increases $\beta$ when sees falling battery voltage (Fig. \ref{results2}(d)) indicating a conservative adaptation. 

Figure \ref{results2}(b),compares the requested current envelope $I_r$ and safety filtered current $I_s$. The safety filter enforces voltage constraints and physical limits as defined in Section 2.3. The figure shows: when voltage margin is sufficient $I_s \approx I_r$, as voltage approaches the lower bound, scaling reduces $I_s$. This behavior confirms that the safety filter operates as an enforcement layer independent of LLM reasoning, even when the LLM recommends aggressive discharge (large $\alpha$), physical constraints are strictly maintained as discussed in section 2.3. Figure \ref{results2}(c) shows the recommended velocity trajectory ($v_r$) and the actual vehicle velocity ($v$). The vehicle tracks the recommended velocity with bounded error. Despite adaptive changes in current envelope, the longitudinal dynamics remain stable and well-behaved. Figure \ref{results2}(d) further shows EV battery current and voltage evolution. As discharge progresses, voltage gradually decreases. Figure \ref{results2} confirms that the LLM–physics integration achieves contextual adaptation without compromising safety or stability.

\begin{figure*}[ht]
    \centering    
    \includegraphics[width=\textwidth,trim={4cm 0 3.8cm 0},clip]{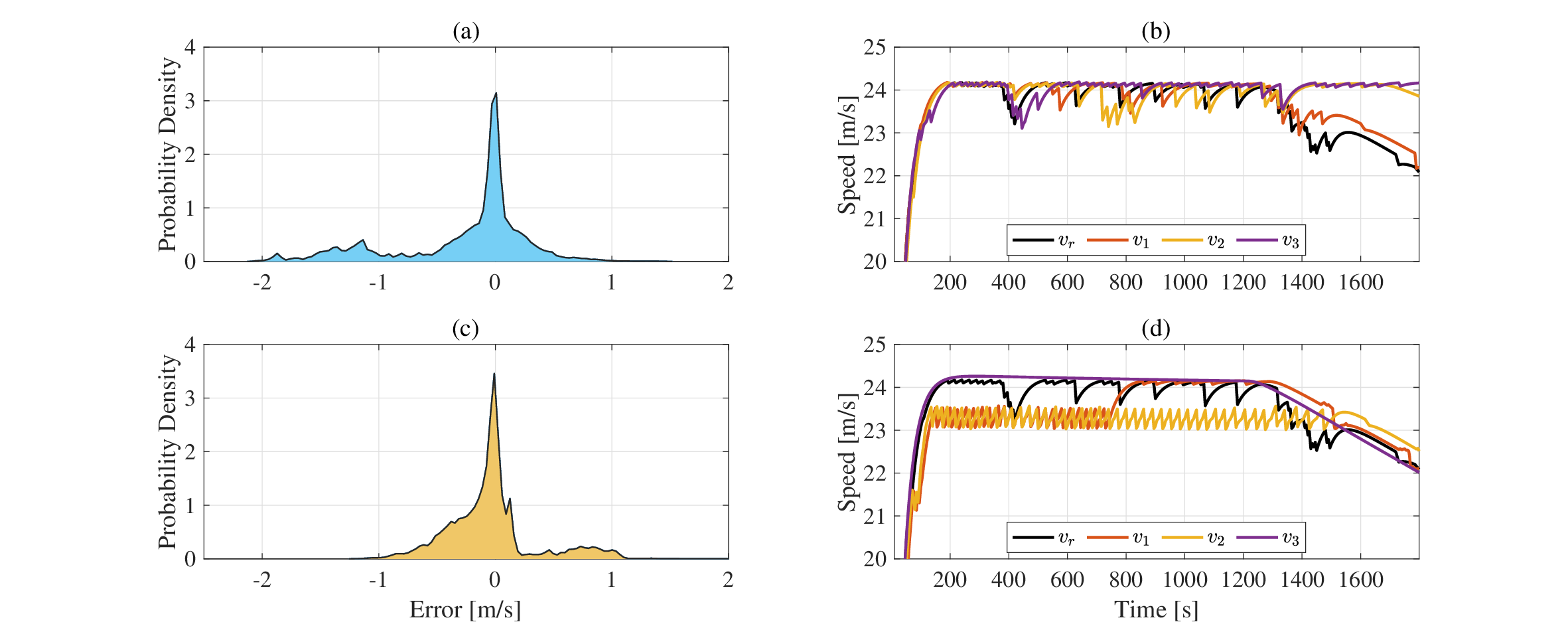}
    \caption{Validation of recommendation error under LLM prompt perturbations and battery model uncertainty for a 30 minutes drive. (a): Probability density of velocity error under battery model perturbations. (b): Comparison of velocity trajectories under recommended ideal battery model ($v_r$) and perturbed ones ($v_1$ with $\delta_A = 0.95\delta_A, \delta_B= 0.95\delta_B$, $v_2$ with $\delta_A = 0.75\delta_A, \delta_B= 0.75\delta_B$, $v_3$ with $\delta_A = 0.5\delta_A, \delta_B=0.5\delta_B$). (c): Probability density function of velocity recommendation error under perturbed LLM prompts. (d): Comparison of recommended velocity trajectories under ideal prompt ($v_r$) and perturbed prompts ($v_1$ with extremely aggressive driving style, $v_2$ with extremely aggressive driving style and frequent acceleration bursts, $v_3$ with a contradictory user prompt as ``I need to reach the store quickly, I’m already late. But I also want to save battery because I might go out again later'').}
    \label{results3}
\end{figure*}

\textbf{Robustness analysis under LLM prompt and battery model inaccuracies:} Next we evaluate how inaccuracies in user prompting and battery modeling propagate through the closed-loop system, as analyzed theoretically in Section 3. To evaluate $\delta_{\lambda}$ we injected mild inconsistencies in user prompts such as: ambiguous driving preference, conflicting urgency statements, incomplete velocity feedback summaries and to evaluate battery model inaccuracies ($\delta_A$ and $\delta_B$) we scale the system matrices in \eqref{ssbatt}. Figure \ref{results3}(a) and (c) shows the probability density of the velocity error distribution under battery model parameter perturbations and  under prompt perturbations respectively. Error distribution is centered near zero, heavy tails remain bounded and no drift or divergence observed. Figure \ref{results3}(b) and (d) shows velocity trajectories under battery parameter perturbations and under multiple perturbed prompts, respectively. The trajectories remain clustered and bounded. These observations align with the inequality \eqref{ly-222}. It also demonstrates that LLM-induced variability does not destabilize the system. To systematically analyze the error between ideal and actual recommendations, different imperfect prompts were added to the nominal prompt. These imperfect prompts are organized into distinct families, summarized in Table \ref{table_prompt_families}.

\begin{table}[h]
\caption{LLM prompt perturbations}
\centering
\begin{tabular}{|l|p{5.8cm}|}
\hline
\textbf{Prompt Family} & \textbf{Semantic Description} \\ 
\hline
Nominal/Ideal & 30-min city trip, mixed driving style \\ 
\hline
Conservative & Smooth driving, preserve battery, etc \\ 
\hline
Aggressive & Running late, fast arrival preferred \\ 
\hline
Ambiguous & Wants to go fast but also drive slowly \\ 
\hline
Layman Error & Informal/vague phrases (e.g., “just go fast”) \\ 
\hline
\end{tabular}
\label{table_prompt_families}
\end{table}

The results demonstrate that the proposed architecture successfully integrates high-level behavioral reasoning (via LLM) with low-level physical enforcement (via physics-driven safety filter). The LLM may influence preference and current envelope shaping, but it cannot directly violate physical constraints. 



\section{Conclusion}
This work presents a user-aware EV driving management framework that unifies LLM-based semantic reasoning with electrochemical battery physics. By structurally separating high-level behavioral inference from low-level physical enforcement, the architecture enables contextual adaptation while guaranteeing voltage safety and stability. We analytically derived bounds on recommendation error under model and prompt uncertainty, and simulations validated bounded, constraint-consistent behavior. The results demonstrate that AI-driven intent interpretation can be safely embedded within physics-based control systems.

\bibliographystyle{ieeetr}
\bibliography{ifacconf}             
                                                   







\end{document}